\documentclass[letterpaper]{article} 
\usepackage{aaai24}  
\usepackage{times}  
\usepackage{helvet}  
\usepackage{courier}  
\usepackage[hyphens]{url}  
\usepackage{graphicx} 
\usepackage{natbib}  
\usepackage{caption} 
\usepackage{algorithm}
\usepackage{algorithmic}
\usepackage{thmtools,thm-restate}
\newcommand{\specialcell}[2][c]{%
  \begin{tabular}[#1]{@{}c@{}}#2\end{tabular}}

\newcommand{\bvec}[1]{\boldsymbol{ #1 }}

\renewcommand{\hat}{\widehat}
\renewcommand{\tilde}{\widetilde}

\def\min{\qopname\relax n{min}}
\def\max{\qopname\relax n{max}}
\def\argmin{\qopname\relax n{argmin}}
\def\argmax{\qopname\relax n{argmax}}

\def\Pr{\qopname\relax n{\mathbf{Pr}}}

\def\X{\mathcal{X}}

\newcommand{\eat}[1]{}

\newenvironment{lp*}{\begin{equation*}  \begin{array}{lll}}{\end{array}\end{equation*}}

\input{style.sty}
\usepackage{newfloat}
\usepackage{listings}
\DeclareCaptionStyle{ruled}{labelfont=normalfont,labelsep=colon,strut=off} 
\floatstyle{ruled}
\newfloat{listing}{tb}{lst}{}
\floatname{listing}{Listing}
\usepackage{amsfonts}
\usepackage{amsmath,amsthm,amssymb}
\usepackage{bm}
\usepackage[short,nocomma]{optidef}

\newcommand{\ma}[1]{{\color{green}  [\text{Michael:} #1]}}

\title{Learning in Online Principal-Agent Interactions: The Power of Menus}
\author{
    Minbiao Han\textsuperscript{\rm 1},
    Michael Albert\textsuperscript{\rm 2},
    Haifeng Xu\textsuperscript{\rm 1}
}
\affiliations{
    \textsuperscript{\rm 1}Department of Computer Science, The University of Chicago\\
    \textsuperscript{\rm 2}Darden Business School, University of Virginia\\
    \{minbiaohan, haifengxu\}@uchicago.edu, albertm@darden.virginia.edu

}

\begin{document}

\maketitle

\begin{abstract}

We study a ubiquitous learning challenge in online principal-agent problems during which the principal learns the agent's private information from the agent's revealed preferences in historical interactions.
This paradigm includes important special cases such as pricing and contract design, which have been widely studied in recent literature.
However, existing work considers the case where the principal can only choose a single strategy at every round to interact with the agent and then observe the agent's revealed preference through their actions.
In this paper, we extend this line of study to allow the principal to offer a \emph{menu} of strategies to the agent and learn additionally from observing the agent's selection from the menu.
We provide a thorough investigation of several online principal-agent problem settings and characterize their sample complexities, accompanied by the corresponding algorithms we have developed.
We instantiate this paradigm to several important design problems --- including Stackelberg (security) games, contract design, and information design.
Finally, we also explore the connection between our findings and existing results about online learning in Stackelberg games, and we offer a solution that can overcome a key hard instance of \citet{peng2019learning}.
\end{abstract}

\section{Introduction}
Asymmetric information is a key friction in many economic interactions. 
When setting a price for an item, the seller typically does not know the exact value the buyer has for the item. 
In hiring and compensation decisions, the employer does not know the work ethic of the particular employee.
For a security professional, the specific value of different potential targets to an attacker may be unknown.
Note that this information asymmetry is distinct from the standard strategic considerations in these interactions. For example, even if the security professional knows the value of all targets, an attacker may choose to attack a low-value target in order to avoid any potential defenses. However, if the security professional does not know the value of the targets to the attacker, it is much less likely that she will be able to implement a successful defense.

Learning from interactions where there is a single \emph{principal} (referred to as she) who acts first and then an \emph{agent} (referred to as he) who responds, a class of problems denoted \emph{principal-agent problems} \cite{myerson1982optimal,gan2022optimal}, has generated significant interest in the computer science literature (e.g., \citealp{kleinberg2003value,letchford2009learning,haghtalab2022learning}).
However, existing work focuses on learning the agent's type by using carefully chosen strategies from the possible set of principal strategies.
In the pricing of a single item example, this would be the principal setting a price in every round of play to narrow down the agent's value for the item.
Specifically, the principal's strategy space consists of offering a fixed price, $p$, and a probability of allocating an item, $x$, to the agent who has an unknown valuation $v$, which is his type.
The agent can choose to purchase the item or not, providing information to the principal about the agent's type.
In this setting, the best strategy for the principal to learn the agent's type exactly requires $O(log(K))$ rounds, where $K$ is the number of possible agent types.



Often, in practice, a principal does not choose a single strategy but instead offers a \emph{menu} of strategies.
The agent then chooses among the strategies, potentially revealing something about their type.
Exploiting \emph{revealed preferences} through offering a menu of strategies is common in pricing.
For example, a retailer may bundle goods at different prices, or there may be different pricing based on the quantity purchased.
If the principal can offer a menu of possible strategies, can she learn significantly more efficiently?

To illustrate the type of results we demonstrate in the revealed preferences setting with menus, consider the pricing problem again.
In this setting, instead of posting a single price and allocation probability every round, the principal posts a menu of prices  $\{x, p(x)\}_{x \in [0,1]}$.
The agent then chooses a strategy, i.e. an $x$, from the menu, at which point the principal plays that strategy.
In this case, the agent's type can be learned in a single round.
Specifically, if the payment function $p(x)$ is strictly convex, then $x^* = \argmax_{x} v\cdot x - p(x)$ is uniquely determined by $v$, and $v = \nabla p(x^*)$.

In this work, we examine, for general principal-agent problems, the question of the sample complexity of identifying the agent's private type.
We provide a condition, which for most common classes of principal-agent problems is \emph{generic}, that allows the principal to identify the agent's type in a \emph{single round}.
Additionally, for settings under which the aforementioned condition does not hold, we demonstrate improvements using menus relative to state of the art single strategy results.
\subsection{Related Work}
Our setting is a generalized principal-agent problem \cite{myerson1982optimal,gan2022optimal}. Online learning in principal-agent problems has been well studied in the context of pricing \cite{kleinberg2003value,amin2013learning,mohri2014optimal,dawkins2021limits}, Stackelberg games \cite{letchford2009learning,haghtalab2022learning}, security games \cite{balcan2014learning,peng2019learning}, contract design problems \cite{ho2014adaptive,zhu2022sample}, and Bayesian persuasion \cite{castiglioni2020online,harris2023algorithmic}. There, the principal aims to learn the optimal strategy against an unknown agent through repeated interactions. The principal achieves this by posting a single principal strategy in each round and observing the corresponding response from the agent. Our results consider the setting where the principal can post (infinitely) many principal strategies. In Section \ref{sec:connection}, we provide direct comparisons to the learning environment of a single strategy, where we solve the previous hardness example by \cite{peng2019learning}.

Another area of relevant literature is the research on learning economic parameters from revealed preferences \cite{beigman2006learning,zadimoghaddam2012efficiently,balcan2014learning,dawkins2022first}. The main difference between this line of research and the previous research about online learning in principal-agent problems is that these works focus on learning the agent's utility parameter, instead of the optimal principal strategy. Similar to us, they also consider learning the agent's utility parameters. 
However, these previous works all focus on designing learning algorithms that post a single principal strategy at every round. In addition, we consider general principal-agent problems and instantiate the results to multiple specific problems, including the specific pricing problem considered in this literature.

More generally, our research aligns with the general field of learning from strategic data sources.  This line of inquiry has been explored in various contexts and domains, addressing diverse objectives such as spam filtering \cite{bruckner2011stackelberg}, classification under incentive-compatibility constraints \cite{zhang2021incentive}, and examining social implications \cite{akyol2016price,milli2019social}. Finally, going beyond learning, our work subscribes to the line of information elicitation, which involves gathering relevant information from individuals to learn efficient algorithms. This literature is extensive and spans various disciplines, such as Bayesian persuasion \cite{dughmi2017algorithmic,kamenica2019bayesian}, decision making \cite{savage1971elicitation,chen2011information}, strategic data collection \cite{goel2020personalized,kong2020information}, and reverse Stackelberg games \cite{groot2012reverse}. 
\section{Preliminaries}\label{sec:prelim}
We consider a generalized principal-agent problem \cite{myerson1982optimal,gan2022optimal}. 
The 
one-round static version of the problem can be characterized as a two-stage game with two players:
the principal  and the agent. The agent possesses a private type $\theta$, which is randomly drawn from a finite set $\Theta$ with a prior probability distribution $\mu \in \Delta(\Theta)$. The principal is unaware of the specific realization of $\theta$. Furthermore, the agent's action is indicated by $j$, which is beyond the direct control of the principal. Let $[n] = \{1, 2, \cdots n\}$ denote the action space of the agent. The utility functions of the principal and the agent are denoted as $U(\cdot)$ and $V(\cdot)$, respectively, where $U/V: \xxx \times [n] \rightarrow \rr$ and $\xxx \subseteq \rr^m$ denotes the strategy space of the principal. Throughout this paper, we assume that the agent's utility function $V(\cdot)$ is linear with respect to the principal strategy $\x \in \xxx$, that is, 
\begin{equation}\label{eq:linear_agent_utility} V(\x, j) = \langle \v_j, \x \rangle + c_j,\end{equation}
where \begin{small}$\langle, \rangle$\end{small} denotes the inner product, $\v_j$ and $c_j$ represent the agent's utility parameters when he takes action $j$.  We note that this assumption is widely adopted by the computer science literate studying principal-agent problems \cite{von2004leadership,dutting2019simple,gan2022optimal}. The utility information of distinct agent types is denoted by the superscript $\theta$, that is, $V^{\theta}(\x, j) = \langle \v_j^\theta, \x \rangle + c_j^\theta$.


Throughout, vectors are denoted by bold lowercase letters (e.g., $\x, \v_j$). The $i^{th}$ component of a vector is denoted in the subscript of a non-bold letter (e.g., $x_i, v_{j,i}$). In addition, we define a special relation between two vectors $\x, \y \in \rr^m$.
\begin{definition}
    Given $\x, \y \in \rr^m$, we denote $\x \parallel \y$ if there exists a scalar $\lambda  \in \rr$ such that $\x = \lambda \y$; otherwise, we denote $\x \nparallel \y$.
\end{definition}
\noindent Note that we always have $x \parallel y$ when $x,y \in \rr$. Next, let us now consider some specific instantiations of the above generalized principal-agent problem framework.

\paragraph{Stackelberg Games.} One widely adopted special case of generalized principal-agent problems is the celebrated model of Stackelberg games \cite{stackelberg1934marktform,von2004leadership,10.1145/3580507.3597680}. In such games, we commonly refer to the principal as the leader and the agent as the follower. The leader has $m$ available actions and the follower has $n$ available actions. The leader's strategy can be represented by a randomized strategy $\x \in \Delta^m = \{\x: \sum_{i \in [m]} x_i = 1\}$, where $x_i$ denotes the probability of playing each pure action $i$.  The follower's response is denoted by action $j\in [n]$. The leader/follower reward information is represented as matrix $L/F\in \rr^{m \times n}$, where each element $L_{i,j}/F_{i,j}$ represents the leader/follower's reward when the leader takes action $i$ and the follower plays action $j$. As a result, we have 
$V^\theta(\x, j) = \sum_{i \in [m]} x_i F^{\theta}_{i,j},$
where $F^\theta$ represents the follower type $\theta$'s corresponding reward matrix. 

\paragraph{Stackelberg Security Games.} The Stackelberg Security Games (SSGs) have been studied for security resource allocation problems in many scenarios such as the Federal Air Marshal Service, the US Coast Guard, and the wildlife protection \citep{tambe2011security,an2017stackelberg}. The principal, known as the defender, has $r$ resources to protect $n$ targets from the attacker. Thus, the principal's defense strategy can be represented by $\x \in \rr^n$ where each $x_t$ represents the probability that target $t \in [n]$ is protected by the defender and $\sum_t x_t \le r$. When a covered target $t \in [n]$ is attacked, the attacker gets penalty $P^a_t$ and the defender gets reward $R^d_t$; If an uncovered target is attacked, the attacker gets reward $R^a_t$ and the defender gets penalty $P^d_t$. We assume the reward is preferred to the penalty, i.e., $R^d_t>P^d_t$ and $R^a_t>P^a_t$ for all $t$. The attacker's utility function is  $ V^\theta(\x, t) =  x_t P^{a, \theta}_t + (1-x_t)R^{a, \theta}_t,$ where each type $\theta$ has a corresponding penalty/reward (i.e., $P^{a, \theta}_t/R^{a, \theta}_t$) for target $t$.
\paragraph{Contract Design Problems.} 
Another instance of the principal-agent problem is the Nobel prize-winning research on contract theory \cite{hart1986theory}. In a contract design problem, the agent has $n$ actions that the principal cannot directly observe or control. Instead, the principal observes the $m$ different possible outcomes and receives the reward $r_i \in \rr$ associated with each outcome $i \in [m]$. We denote $\r \in \rr^m$ as the reward vector for all outcomes. Each action of the agent $j \in [m]$ induces a certain distribution of outcomes $\p_j \in \Delta^m$ at cost $c_j$. The principal's strategy is to design a contract $\x \in \rr^m$ that specifies the payment based on each possible outcome to incentivize the agent to take certain actions. As a result, we have 
$V^\theta(\x, j) = \sum_{i \in [m]}x_i p^\theta_{j, i} - c^\theta_j,$
where each agent type $\theta$ has a corresponding  cost $c^\theta_j$ for playing action $j$ with a realization probability $\p^\theta_j$ over the $m$ outcomes. 

\paragraph{Information Acquisition Games.} 
In an information acquisition problem \cite{savage1971elicitation,10.1145/3490486.3538338,chen2023learning}, we consider a stochastic environment with a principal and an agent.  There exists a hidden state $w \in \Omega$ that influences the principal's utility but remains unknown to both the agent and the principal until the interaction ends. We denote $\Delta(\Omega)$ as a probabilistic belief over the hidden state. The principal initiates the process by offering a \textit{scoring rule} $S: \Delta(\Omega) \times \Omega \rightarrow \rr$ to the agent. The agent chooses an action $j \in [n]$ with a cost $c_j$ and receives an observation  $o \in O$ related to the hidden state $w$ with probability $\Pr(o, w|j)$.  They then provide a refined belief report $\sigma^o \in \Delta(\Omega)$ via the Bayesian rule (i.e., $\sigma^o(w) = \Pr(o, w|j) / \Pr(o|j)$), enabling the principal to make an informed decision. Finally, the hidden state $w$ is revealed, and the principal determines the agent's payment based on the scoring rule. As a result, we have $V^\theta(\x, j) = \sum_{w,o} S(\sigma^o, w){\textstyle\Pr^{\theta}}(w, o|j) - c^\theta_j,$ where $\sigma^o \in \Delta(\Omega)$, $w \in \Omega$, $\Pr^{\theta}(w, o|j)$ represents the agent $\theta$'s probability of observing $o$ if the hidden state is $w$ and the action taken is $j$, and $c^\theta_j$ represents the corresponding cost.

\section{When is a Single Round Sufficient?}\label{sec:single_query}
In this section, we demonstrate the feasibility of learning the agent's private type within a single round, given the following assumption on the agent's finite type set $\Theta$.


\begin{assumption}\label{assum:non_proportional}
    Given a finite-type set $\Theta$ and $\theta, \theta' \in \Theta$ such that $\theta \ne \theta'$, we assume that $\v^\theta_j \nparallel \v^{\theta'}_j$ for every action $j \in [n]$, where $\v^\theta_j = \nabla V^\theta(\x, j)$. 
\end{assumption}
\noindent Note that this assumption is one of nondegeneracy. 
To be precise, each principal-agent problem instance is specified by the principal and agent utility parameters. This can be thought of as a point in a high-dimensional space, one dimension for every parameter. As a result, we can think of every point in this high-dimensional space being a problem instance. The assumption holds with probability $1$ since the game instances that violate the assumption in this space form a zero-measure set.
Given this assumption, we are prepared to present the key findings of this section.

\begin{restatable}{restatethm}{constantQuery}\label{thm:constant_query}
Suppose $\Theta$ satisfies assumption \ref{assum:non_proportional}, 
we can construct a menu of principal strategies $\mmm \subseteq \xxx$ that learns the agent's private type $\theta \in \Theta$ in a single round.
\end{restatable}

\begin{proof}
For every agent type $\theta$, we can always separate the principal's strategy space $\xxx$ into at most $n$ sub-regions, where each sub-region consists of principal strategies that have the same agent best response (i.e., the agent response that maximizes his utility under the principal's strategy).

\begin{figure}[tbh]
    \centering
    \includegraphics[width=0.45\textwidth]{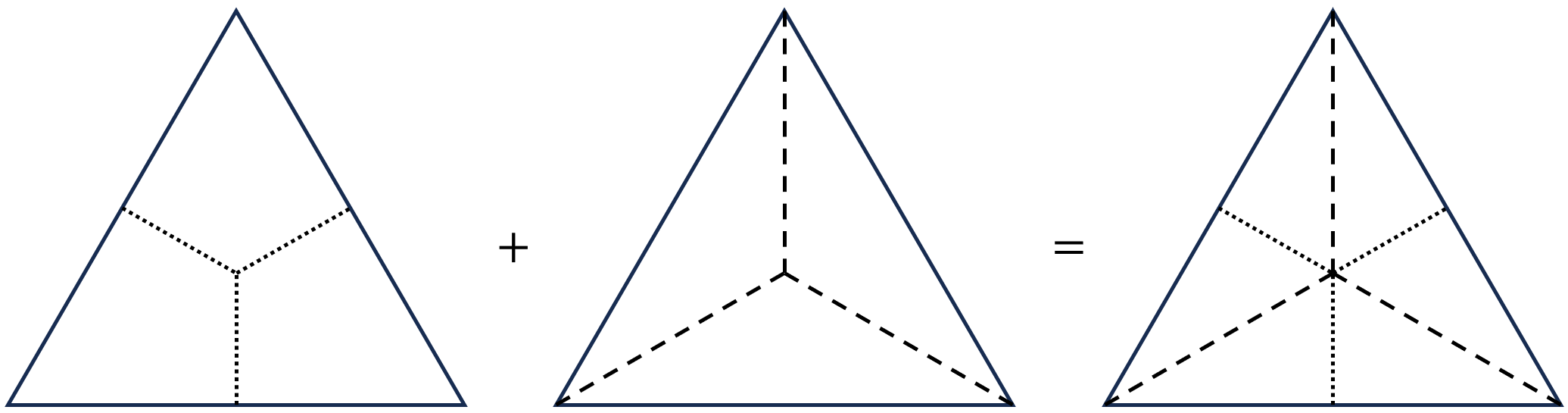}
    \caption{An example where $|\Theta|=2$, $\xxx \subseteq \rr^2$, and $n=3$. For each specific agent type, the principal's strategy space can be divided into three sub-regions, where the agent's best response action remains the same at each sub-region. In this example, there exist six sub-regions where \textit{all} agent types' best responses are fixed for  principal strategies in each sub-region, though those fixed best responses could be different across different agent types.}
    \label{fig:feasible_region}
\end{figure}
\vspace{-0.151mm}
Given a finite set $\Theta$, we can pre-compute the separating hyper-planes from all types as in the example in Figure \ref{fig:feasible_region}.
There always exists a principal strategy $\hat{\x}$ that is an \textit{interior} point such that all agent types' best responses are fixed around this principal strategy.
Then the principal can post a menu of strategies centered at $\hat{\x}$ with a small radius $\epsilon$ (i.e., $\mmm = \{\x: |\x - \hat{\x}|^2\leq \epsilon \} \subset \xxx$) such that all $\x \in \mmm$ have the same best response for every agent type. 
Such a menu $\mmm$ always exists given a small enough $\epsilon$ and that $\hat{\x}$ is the \textit{interior} point of some sub-region.
Suppose the agent chooses $\x^* \in \mmm$ and responds with action $j$.

First, we can exclude those agent types whose best response with respect to $\hat{\x}$ is not $j$. As for the remaining possible types, we can simulate the agent's choice of $\x^* \in \mmm$ through the following optimization program.
\begin{equation}\label{op:simplex_principal_specific_response}
\begin{split}
    \textstyle \max_{\x \in \mmm} \quad &V^\theta(\x, j)  \\
     \text{s.t.}  \quad & \textstyle \sum_i (x_i - \hat{x}_i)^2 \le \epsilon.
\end{split}
\end{equation}
Let us write the Lagrangian function of program \eqref{op:simplex_principal_specific_response} as $\lll(\x, \lambda, \mu) = V^\theta(\x, j) + \lambda \big(\epsilon - \sum_i (x_i - \hat{x}_i)^2\big) $ 
and its corresponding Karush–Kuhn–Tucker (KKT) conditions:
{\allowdisplaybreaks
\begin{align*}
        &v^\theta_{j,i} - 2\lambda x^*_i + 2\lambda \hat{x}_i = 0, \quad \forall i \in [m]\\
        & \lambda(\epsilon - \textstyle \sum_i (x^*_i-\hat{x}_i)^2) = 0,\\
        & \epsilon \ge \textstyle  \sum_i (x^*_i-\hat{x}_i)^2, \\
        & \x^*, \lambda \ge 0.
\end{align*}}
By solving the above KKT conditions, we have $x^*_i(\theta) = \frac{\sqrt{\epsilon} v^\theta_{j,i}}{\|\v^\theta_j\|_2} + \hat{x}_i, \text{ for all } i$. 
According to assumption \ref{assum:non_proportional}, it follows that $\v^\theta_j \nparallel \v^{\theta'}_j$, resulting in $\x^*(\theta) \ne \x^*(\theta')$ for all $\theta, \theta' \in \Theta$ and $\theta \ne \theta'$. 
As a result, we can learn the agent's type $\theta \in \Theta$ in a single round since each agent type $\theta$ has a unique choice of $\x^*(\theta)$ given our menu of principal strategies $\mmm$.
\end{proof}
We remark that assumption \ref{assum:non_proportional} does not hold when $m=1$, since we always have $v^\theta_j \parallel v^{\theta'}_j$ when $v^{\theta}_j, v^{\theta'}_j \in \rr$.
Thus, the menu of principal strategies $\mmm$ we constructed for Theorem \ref{thm:constant_query} does not hold for $m = 1$. 
Specifically, we have $x^*(\theta) = \sqrt{\epsilon} v^\theta_j/|v^\theta_j| + \hat{x} \in \{\hat{x} +\sqrt{\epsilon}, \hat{x}- \sqrt{\epsilon}\}$, which does not depend on $\theta$, when $m=1$. As a result, when $|\Theta| > 2$, there must exist two agent types that share the same preference for $x^*$ when presented with the given menu of strategies $\mmm$, because $x^*(\theta)$ only has two possible choices. 
In Section \ref{sec:finite_log}, we will introduce effective algorithms that address this challenge and facilitate the learning of the agent's type specifically in scenarios where the principal action space is limited to one dimension. Next, we apply the generalized principal-agent framework to specific principal-agent problems.

\subsection{Stackelberg Games}\label{subsec:stackelberg}

Our first instantiation is the  Stackelberg game, where we show a corollary of Theorem \ref{thm:constant_query}. 
According to the definition of follower utility in a Stackelberg game, we have $v^\theta_{j, i} = \nabla V^\theta_i(\hat{\x}, j) = F^\theta_{i, j}$. 
\begin{corollary}\label{corollary:stackelberg}
    Suppose $\Theta$ satisfies assumption \ref{assum:non_proportional}, 
    we can construct a menu of leader strategies $\mmm \subseteq \xxx$ that learns the follower's private type $\theta \in \Theta$ in a single round.
\end{corollary}
The proof is straightforward based on the proof of Theorem \ref{thm:constant_query}. However, a special aspect of Stackelberg games is that when $m = 2$, the leader's strategy space is $\xxx = \{\x: x_1 + x_2 = 1\} \subseteq \rr$ and the follower's utility function is $v^\theta(\x, j) = (v^\theta_{j, 1} - v^\theta_{j,2}) x_1 + v^\theta_{j,2}$. In other words, the actual dimension of the principal's strategy space is $1$ when $m=2$ in a Stackelberg game. In this case, Theorem \ref{thm:constant_query} does not hold as we discussed before, and it is not possible to learn the follower's private type within a single round. 

\begin{restatable}[Impossibility of Single-round Learning]{restateprop}{impossibleLeader}\label{prop:principal_two_actions}
There exists a finite set of follower types $\Theta$ such that it is impossible to learn the follower's type within a single round when the leader has $2$ actions and $|\Theta| > 2$.
\end{restatable}
\begin{proof}
We prove this by contradiction. Suppose the agent has $K$ possible types (i.e., $\Theta = \{\theta_1,\cdots, \theta_K\}$) and we can construct a menu of leader strategies $\mmm$ such that there exists $K$ unique $\{x^{\theta_k}\}_{\theta_k \in \Theta}$ where $x^{\theta_k} = \argmax_{\x \in \mmm} V^{\theta_k}(\x)$. 
The agent's utility function $V^{\theta_k}(x) = \max_j V^{\theta_k}(x, j)$ is a piece-wise linear strictly convex function. By convexity, we have $\argmax_{x \in \mmm} V^{\theta_k}(x) \in \{\min(\mmm), \max(\mmm)\}$, for all type $\theta_k$, contradicting the fact there exists $k>2$ unique choices of all follower types.
\end{proof}

\subsection{Contract Design Problems}\label{subsec:contract}
We also consider the special contract design problem where $v^\theta_{j, i} = \nabla V^\theta_i(\x, j) = p^\theta_{j, i}$, which represents the probability to achieve outcome $i \in [m]$ when the agent takes action $j \in [n]$. 
Next, we apply Theorem \ref{thm:constant_query} and show a corollary result for the contract design problem. 
\begin{corollary}\label{corollary:contract}
    Suppose $\Theta$ satisfies assumption \ref{assum:non_proportional}, 
    we can construct a menu of principal contracts $\mmm \subseteq \xxx$ that learns the agent's private type $\theta \in \Theta$ in a single round.
\end{corollary}

\subsection{Information Acquisition Games}
Another special case of the generalized principal-agent problem is the information acquisition problem. 
From the definition of  agent utility function, we have $v^\theta_{j, i} = \nabla V^\theta_i(\x, j) = \Pr^{\theta}(w, o|j)$ where $i$ denotes a pair of variables  $(w, o), w \in \Omega, o \in O$. Thus, $\v^\theta_j$ in the information acquisition game represents a vector of length $|\Omega| \times |O|$. We can also apply Theorem \ref{thm:constant_query} for the following corollary.

\begin{corollary}\label{cor:info_acquisition}
    Suppose $\Theta$ satisfies assumption \ref{assum:non_proportional}, 
    we can construct a menu of principal scoring rules $\mmm \subseteq \xxx$ that learns the agent's private type $\theta \in \Theta$ in a single round.
\end{corollary}

\subsection{Stackelberg Security Games}\label{subsec:security_games} 
Finally, we consider the specific Stackelberg security game.
Recall that the attacker's utility to attack target $j$ only depends on the defender's coverage probability of protecting target $j$, that is, \begin{equation}\label{eq:security_attacker_utility}\textstyle V^\theta(\x, j) = P^{a,\theta}_{j} x_j + (1-x_j) R^{a,\theta}_j.\end{equation}
As a result, the actual dimension of the defender's strategy space in $V^\theta(\x, j)$ is $1$ (i.e., $x_j$ the probability of defending target $j$), and the result from Theorem $\ref{thm:constant_query}$ cannot be applied. 
Intuitively, when presented with a menu of defender strategies, all agent types would favor the one with the lowest protection probability for target $j$, given that their best response with respect to all strategies from this menu is fixed to attack this particular target $j$.
The underlying reason is $v^\theta_{j, i} = \nabla V^\theta(\x, j) = P^{a,\theta}_{j} - R^{a,\theta}_j \in \rr$. As a result, assumption \ref{assum:non_proportional} cannot hold for more than two attacker types in a security game, and Theorem \ref{thm:constant_query} is also not applicable.


\section{When is a  Small Number  of Rounds Sufficient?}\label{sec:finite_log}
As discussed, our single-round results do not hold when the effective dimensionality of the principal's strategy space is $1$ (see Proposition \ref{prop:principal_two_actions}). In the following section, we introduce algorithms designed to learn the private agent type efficiently with a small number of rounds in this particular scenario, i.e., $\xxx \subseteq \rr$. For the sake of simplicity, we can assume that $\xxx = [0, 1]$, as the principal's strategy can always be rescaled within this range without loss of generality.

First, we present an algorithm for the principal to learn the agent's type within $\log |\Theta|$ rounds under the following assumption on the finite type set $\Theta$. 
\begin{assumption}\label{assum:menu}
    Given a finite type set $\Theta$, we assume for every type $\theta \in \Theta$, there does not exist $j \in [n]$ such that
    \begin{equation}
        V^\theta(\x, j) \ge V^{\theta}(\x, j'), \quad \forall \x, j' \ne j
    \end{equation}
\end{assumption}
\noindent Assumption \ref{assum:menu} is a statement that no agent type has a  dominant action.
This ensures there is sufficient variation in the agent's behavior for the principal to learn the agent's type.



\begin{restatable}{restatethm}{smallQueryMenu}\label{thm:logk_without_agent_response}
When $m=1$ and $\Theta$ satisfies assumption \ref{assum:menu}, there exists an algorithm to learn the agent's type in $\log |\Theta|$ rounds where the principal posts a menu of principal strategies per round.
\end{restatable}
\begin{proof}
We propose the following algorithm \ref{alg:learning_menu} to learn the agent's true type $\theta \in \Theta$ in $\log |\Theta|$ rounds. 
    \begin{algorithm}[tbh]
    \caption{\textsc{Learning-Via-Menu}}
    \label{alg:learning_menu}
    \textbf{Input}: Agent's type set: $\Theta = \{\theta_1, \cdots, \theta_K\}$\\
    \textbf{Output}: Agent's true type $\theta$
    \begin{algorithmic}[1] 
    \STATE Let $V^{\theta_k}(x) = \max_j V^{\theta_k}(x, j), \forall \theta_k \in \Theta$.
    \STATE Compute $x^{\theta_k} = \argmin_{x}V^{\theta_k}(x), \forall \theta_k \in \Theta$, let $\bvec{x} = \text{\textbf{sorted}}([x^{\theta_1}, \cdots, x^{\theta_K}])$
    \WHILE{len($\bvec{x}$) $>$ 1:}
    \STATE $n$ = len($\bvec{x}$); $mid = \lfloor \frac{n}{2} \rfloor$; $x_1, x_2 =\bvec{x}[mid], \bvec{x}[mid+1]$ 
    \STATE Provide a menu $\mathcal{M}=(x_1, x_2)$ of two principal strategies to the agent
    \IF {the agent chooses $x_1$}
    \STATE $\bvec{x} = \bvec{x}[mid+1:]$ (deleting all the elements $x[0],\cdots, x[mid]$ from $\bvec{x}$)
    \ELSE
    \STATE $\bvec{x} = \bvec{x}[:mid]$ (deleting all the elements $x[mid+1],\cdots, x[n-1]$ from $\bvec{x}$)
    \ENDIF
    \ENDWHILE
    \STATE \textbf{return} the only $x^{\theta_k}$ in $\bvec{x}$, where $\theta_k$ is the agent's type
    \end{algorithmic}
    \end{algorithm}

The correctness of the proof of this algorithm relies on a key insight that $V^{\theta_k}(x)$ is a convex function since it is the maximum over a set of linear functions. Thus, given two principal strategies $x_1$, $x_2$ where $x_1 < x_2$, if the agent $\theta$ prefers $x_1$, then the agent's $x^\theta = \argmin_x V^{\theta}(x)$ must be on the right of $x_1$ by the convex property of agent's utility function. A similar argument holds when the agents prefer $x_2$. As a result, in each round we can eliminate half of the types from the agent's type set, leading to a $\log |\Theta|$ sample complexity. 

However, there is still one more case we need to consider. When two agent types $\theta_k, \theta_{k'}$ have the same $x^{\theta_k} = x^{\theta_{k'}}$ in Line 2, they will always choose the same strategy for any menu of two strategies from Line 4. In this case, the principal can distinguish these two agent types by observing the agent's corresponding response to the chosen principal's strategy. 
By assumption \ref{assum:menu} there does not exist a dominant action for any agent type. Therefore, every agent type has more than one best response at $x^{\theta_k}$. We follow the convention that when the agent is indifferent between multiple responses, we can specify the tie-breaking rule \cite{letchford2009learning,peng2019learning}, ensuring that different types respond with different actions. Thus, the principal can still distinguish these different agent types according to different responses specified by different tie-breaking rules. 
\end{proof}

What is more, we also consider the scenario where the principal cannot offer a menu of principal strategies, but only a \textit{single} principal strategy. This setting has been extensively studied in the literature where $m > 1$ \cite{letchford2009learning, blum2014learning, peng2019learning}. Next, we show an efficient algorithm when $m=1$, which can learn the agent's private type in $\log |\Theta|$ rounds as our previous results. 
Though our next algorithm achieves the same $\log |\Theta|$ complexity as our previous result in Theorem \ref{thm:logk_without_agent_response}, it requires the following additional assumption on the finite agent type set $\Theta$ since the principal now has less flexibility and can only post a \textit{single} strategy per round.
\begin{assumption}\label{assum:strategy}
    Given a finite type set $\Theta$, we assume there does \textit{not} exist $\theta, \theta'\in \Theta$, $x \in \xxx$ and $j, j^{'}, j^{''} \in [n]$ such that
    \begin{equation}
        V^\theta(x, j) =  V^\theta(x, j') \text{ and } V^{\theta'}(x, j) =  V^{\theta'}(x, j^{''}) 
    \end{equation}
where $\theta \ne \theta'$, $j \ne j'$, and $j \ne j^{''}$.
\end{assumption}

At a high level, assumption \ref{assum:strategy} requires that there do not exist two agent types who not only have the same breakpoint in their piece-wise linear convex utility functions but also change from the same best response to different best responses at this breakpoint. 
In Figure \ref{fig:assumption}, we provide an illustration of the utility functions associated with two agent types. It is important to observe that instances failing to satisfy assumption \ref{assum:strategy} require not only the breakpoints of two piece-wise linear convex functions to coincide, which forms a zero measure set in the entire strategy space, but also necessitate the transition from the same response to different responses for these piecewise linear functions (e.g., $j^\theta_1 = j^{\theta'}_1$ in Figure \ref{fig:assumption}). Hence, we emphasize that assumption \ref{assum:strategy} is a nondegeneracy assumption since any randomly generated instance will satisfy this assumption with probability $1$. 

\begin{figure}[tbh]
    \centering
	\includegraphics[width=.285\textwidth]{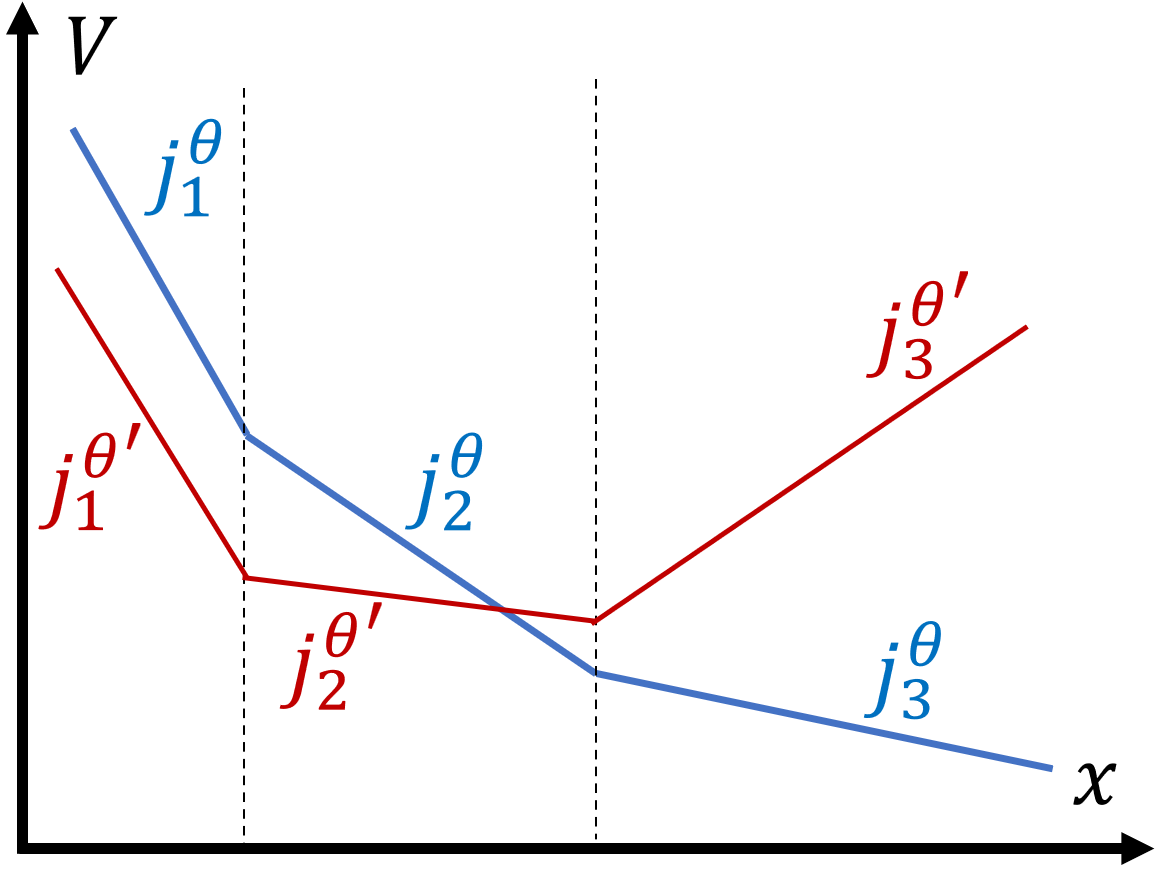} 
	\caption{An example of agents' utility functions $V^\theta$ and $V^{\theta'}$ when principal strategy $x \in \rr$.}
	\label{fig:assumption}
 \vspace{-1mm}
\end{figure}

\begin{restatable}{restatethm}{smallQueryStrategy}\label{thm:logk_with_agent_response}
When $m=1$ and $\Theta$ satisfies assumptions \ref{assum:menu}  \& \ref{assum:strategy},  there exists an algorithm to learn the agent's type in $\log |\Theta|$ rounds where the principal only posts a single strategy per round. 
\end{restatable}
\begin{proof}
We propose the following algorithm \ref{alg:learning_strategy} to learn the agent's true type in $\log |\Theta|$ rounds.

\begin{algorithm}[tbh]
    \caption{\textsc{Learning-Via-Single-Strategy}}
    \label{alg:learning_strategy}
    \textbf{Input}: Agent's type set: $\Theta = \{\theta_1, \cdots, \theta_K\}$\\
    \textbf{Output}: Agent's true type $\theta$
    \begin{algorithmic}[1] 
    \STATE Let $V^{\theta_k}(x) = \max_j V^{\theta_k}(x, j), \forall \theta_k \in \Theta$.
    \WHILE{$|\Theta|> 1$:}
    \STATE For all $j \in [n]$, compute $n^j_{x=0}$ as the number of agent types who best respond to $x=0$ with $j$.
    \STATE Denote $n^{j^*}_{x=0} = \max_j n^j_{x=0}$.
    \IF {$n^{j^*}_{x=0} > \lfloor \frac{|\Theta|}{2} \rfloor$}
    \STATE Compute $x^*\in (0,1]$ such that $n^{j^*}_{x=x^*} = \lfloor \frac{|\Theta|}{2} \rfloor$ 
    \STATE Principal plays strategy $x^*$:
    \IF {Agent responds with $j^*$}
    \STATE Remove agent types whose best response to $x^*$ is not $j^*$ from $\Theta$
    \ELSE
    \STATE Remove agent types whose best response to $x^*$ is $j^*$ from $\Theta$ 
    \ENDIF
    \ELSE
    \STATE Principal plays strategy $x^*=0$:
    \IF {Agent responds with $j^*$}
    \STATE Remove agent types whose best response to $x^*$ is not $j^*$ from $\Theta$
    \ELSE
    \STATE Agent responds with $j'$
    \STATE Remove agent types whose best response to $x^*$ is not $j'$ from $\Theta$
    \ENDIF
    \ENDIF
    \ENDWHILE
    \STATE \textbf{return} $\Theta$
    \end{algorithmic}
    \end{algorithm}
\noindent The proof of this algorithm relies on a key insight that when $n^{j^*}_{x=0} > \frac{|\Theta|}{2}$, we can always find a principal strategy $x^*$ such that the number of agent types responding to $x^*$ with $j^*$ is $\lfloor \frac{|\Theta|}{2} \rfloor$ (i.e. Line 6). The reason is as follows. Given any agent action $j \in [n]$, we denote $n^j_{x=0}$ and  $n^j_{x=1}$ as the number of agent types who respond to $x=0$ and $x=1$ (i.e., the left-most and right-most principal strategy in principal's strategy space) with action $j$. 

According to assumption \ref{assum:menu}, there does not exist a dominated action for every agent type $\theta$. Thus, we must have $n^j_{x=0} + n^j_{x=1} \leq |\Theta|$. In addition, by assumption \ref{assum:strategy}, there do not exist two agent types that change from the same agent response to different responses at the same principal strategy.  As a result, we must have $n^j_{x=0}$  always changes smoothly (i.e., changes at most by $1$) to $n^j_{x=1}$.
Therefore, if $n^{j^*}_{x=0} > \lceil \frac{|\Theta|}{2}  \rceil$ in Algorithm \ref{alg:learning_strategy}, then we must have  $n^{j^*}_{x=1} < \lfloor \frac{|\Theta|}{2} \rfloor$ and
there must exist $x^* \in (0,1)$ such that the number of agent types responding to $x^*$ with $j^*$ is exactly $\lfloor \frac{|\Theta|}{2} \rfloor$.

On the other hand (i.e., Line 13), if $n^{j^*}_{x=0} \le \lceil \frac{|\Theta|}{2} \rceil$, then we know there already exists such an $x^* = 0$. This is due to the fact that for every possible agent response action $j \in [n]$, the number of agent types that responds to $x^*=0$ with this action is less than or equal to $\lceil \frac{|\Theta|}{2} \rceil$. Thus, we can remove at least $\lceil \frac{|\Theta|}{2} \rceil$  possible types no matter which response is observed when playing principal strategy $x^*=0$.

As a result, each round in the \textbf{while} loop (line 2) is guaranteed to remove at least $\lfloor \frac{|\Theta|}{2} \rfloor$  possible agent types from the $\Theta$ set, leading to a $\log |\Theta|$ rounds learning algorithm. 
\end{proof}

Finally, as we discussed in Section \ref{subsec:security_games}, the actual dimension of the defender's strategy space is $1$ in Stackelberg security games. Similarly, we can apply Algorithm \ref{alg:learning_strategy} to learn the attacker's private type in $\log |\Theta|$ rounds. 

\begin{corollary}\label{col:security_log}
Suppose $\Theta$ satisfies assumptions \ref{assum:menu}  \& \ref{assum:strategy},  there exists an algorithm to learn the attacker's private type in $\log |\Theta|$ rounds. 
\end{corollary}

\noindent The proof is straightforward based on Theorem \ref{thm:logk_with_agent_response} and Algorithm \ref{alg:learning_strategy}. The role of the defender and attacker corresponds to the principal and agent in the Stackelberg game. Similarly, Algorithm \ref{alg:learning_strategy} can be applied to security games.

\section{When are Polynomially Many Rounds Sufficient?}\label{sec:infinite_type}
In this section, we consider the more general case where the principal does not have prior information (i.e., $\Theta$) about the agent's private type, and this private type can be drawn from an \textit{infinite} type space (see \citet{letchford2009learning,blum2014learning,peng2019learning} for previous work without menus). 
Following a similar argument to Theorem \ref{thm:constant_query}, we can construct a menu of strategies centered at $\hat{\x}$ with a small radius $\epsilon$ (i.e., $\mmm = \{\x: |\x - \hat{\x}|^2\leq \epsilon \}$). Agent $\theta$'s best response to all principal strategies within this menu is $j \in [n]$ and we denote the agent's optimal choice as $\x^*(\theta)$. From the analysis of Theorem \ref{thm:constant_query}, we have 
\begin{small}\begin{equation}\label{eq:recall_theorem}
    x^*_i(\theta) = \sqrt{\epsilon} v^\theta_{j,i}/ \|\v^\theta_j\|_2 + \hat{x}_i, \quad \forall i \in [m].
\end{equation}\end{small}
As a result, from equation \eqref{eq:recall_theorem} we have
\begin{small}\begin{equation}\label{eq:rewrite_gradient}
    v^\theta_{j,i}/ \|\v^\theta_j\|_2 = \big(x^*_i(\theta) - \hat{x}_i \big)/\sqrt{\epsilon}, \quad \forall i \in [m].
\end{equation}\end{small}
From the above $m$ equations $\eqref{eq:rewrite_gradient}$, we can compute $\v^\theta_j$ up to a factor of $\lambda_j \in \rr$, i.e., $\v^\theta_j = \lambda_j \tilde{\v}^\theta_j$ where $\tilde{\v}^\theta_j$ is an \textit{unit vector} that satisfies \eqref{eq:rewrite_gradient}. As a result, learning the private agent utility parameters $\v^\theta_j \in \rr^m$ is reduced to learning $\lambda_j \in \rr$.

Before proceeding to our main result, we assume that there is an oracle that provides the principal with a set of strategies that induce the agent to best respond with all possible actions for some strategy in the set.
\begin{assumption}[Agent Action-informed Oracle]\label{assum:oracle}
    There exists an oracle that provides $n$ principal strategies $\{\hat{\x}_j\}_{j\in [n]}$ such that under principal strategy $\hat{\x}_j$, the agent's best response is action $j$.
\end{assumption}
\noindent This assumption has been adopted by previous work \cite{chen2023learning,pmlr-v202-zhao23o}. It is considerably less stringent compared to existing online learning models in strategic environments, which assumes that the principal can anticipate the agent's best response or possesses knowledge of certain parameters of the agent's utility function \cite{chajewska2001learning,cesa2006prediction,shalev2012online}. In addition, this may be a reasonable assumption in practice for certain classes of principal-agent problems. For example, in the contract design problem, each outcome may be achieved primarily by a specific action (e.g., a high-quality outcome is most likely achieved by an action with high effort). Therefore, we can induce one specific action by rewarding its corresponding outcome while setting the reward for other outcomes as $0$.
With assumption \ref{assum:oracle}, we are now ready to present the main result of this section.
\begin{restatable}{restatethm}{infiniteType}\label{thm:infinite_type}
Under assumption \ref{assum:oracle}, 
    the agent's private type can be learned 
    with $O(n^2L)$ rounds, where $L$ is the representation precision. 
\end{restatable}
\begin{proof}[Proof Sketch]
The high-level proof idea is as follows. Given any menu $\mmm_j = \{\x: |\x - \hat{\x}_j|^2\leq \epsilon \}$ for all $j \in [n]$, we can learn the agent's utility parameters $\v^\theta_j$ up to a factor of $\lambda_j$, that is, $\v^\theta_j = \lambda_j \tilde{\v}^\theta_j$ where $\tilde{\v}^\theta_j$ is an unit vector that can be computed by \eqref{eq:rewrite_gradient}. Therefore, learning the agent's utility function (i.e., $\v^\theta_j, c_j$ for all  $j \in [n]$) is reduced to learning $\lambda_j, c_j$ for all $j \in [n]$. Furthermore, for any principal strategy $\hat{\x}_{j, j'}$ that's on the separating hyperplane between the best response sub-regions of two actions $j,j' \in [n]$ (see our Figure \ref{fig:feasible_region} for an example of dividing the principal's strategy space into multiple best response sub-regions), we have the agent's utility of responding with $j$ equals to the utility of responding with $j'$ under the principal strategy $\hat{\x}_{j, j'}$, that is,
\begin{small}\begin{equation}\label{eq:hyperplane_main}
\lambda_j \langle \tilde{\v}_j, \hat{\x}_{j, j'} \rangle + c_j = \lambda_{j'} \langle \tilde{\v}_{j'}, \hat{\x}_{j, j'} \rangle + c_{j'}.
\end{equation}\end{small}
where $\lambda_j, \lambda_{j'}, c_{j}$, and $c_{j'}$ are unknown variables. Thus, sampling $4$ principal strategies on this hyperplane is enough to learn these four variables. We continue this process for all possible hyperplanes, and there are at most $n^2$ in total. Detailed proof can be seen in Appendix \ref{appendix:proof}.
\end{proof}

\noindent It is still an open question whether it is sufficient to learn with fewer than $O(n^2 L)$ rounds. 
Using menus, our sample complexity of $O(n^2 L)$ is an improvement over the current state of the art, $O(n^3 L)$ \cite{peng2019learning}, for settings where the principal can only offer a single strategy, an improvement of $O(n)$.
\subsection{Online Learning in Stackelberg Games}\label{sec:connection}
In this subsection, we demonstrate that the ability to offer a menu can dramatically improve the sample complexity of a well-known hardness example from \citet{peng2019learning} for Stackelberg games.
To begin, we provide a summary of current results on Stackelberg games as in Table \ref{table:Stackelberg}.
\begin{table}[tbh]
\centering
    \begin{tabular}{|c|c|c|}\hline
     \specialcell{Number  of \\ rounds}& $m=2$ & $m>2$ \\[0.5em] \hline 
     Menu & \specialcell{$\log |\Theta|$\\ Theorem \ref{thm:logk_without_agent_response}} & \specialcell{ $O(1)$ \\ Corollary \ref{corollary:stackelberg}}\\[0.5em] \hline
     Single strategy & \specialcell{ $\log |\Theta|$\\ Theorem \ref{thm:logk_with_agent_response}} & \specialcell{$|\Theta|$ \\ \cite{peng2019learning}} \\[0.5em] \hline
    \end{tabular}
    \caption{Summary of the results discussed regarding the required number of interaction rounds to learn the private follower's type from a finite type set $\Theta$ in Stackelberg games.} 
    \label{table:Stackelberg}
    \vspace{-5mm}
\end{table}

\begin{example}[Lemma 8, \cite{peng2019learning}]\label{example:peng2019}
    Consider the Stackelberg game, where the leader has $m$ actions and the follower has $n=m+2$ actions. The utility parameters of every follower type $\theta$ is
    \begin{small}
    \begin{equation*}
    \begin{split}
        &F^\theta_{\cdot,j} = (\underbrace{-\tfrac{2}{m-2}, \cdots, -\tfrac{2}{m-2}}_{i-1}, 1, \underbrace{-\tfrac{2}{m-2}, \cdots, -\tfrac{2}{m-2}}_{m-i}), \,\, \forall j \in [m];  \\
        &F^\theta_{\cdot,m+1} = \Big(\underbrace{-\frac{1}{N^3}, -\frac{1}{N^3}, \cdots, -\frac{1}{N^3}}_{m}\Big);  \\
        &F^\theta_{\cdot, a^*} = \frac{1}{N^2} \v^\theta.
    \end{split}
    \end{equation*}
    \end{small}
  
\noindent where $\theta \subseteq [m]$ with $|\theta| = m/2$ or $0$, that is, $\theta$ is either an empty set or a set with $m/2$ elements that are drawn from set $[m]=\{1, \cdots, m\}$. $\v^\theta \in \{1, -N\}^m$ is a vector of length $m$ whose $i$-th element is $1$ if $i \in \theta$ and $-N$ otherwise. There are $\big(\begin{smallmatrix}m\\m/2\end{smallmatrix}\big)+1$ follower types in total. 
\end{example}
In the case of this Example \ref{example:peng2019} from \citet{peng2019learning}, they show that a minimum of $2^{\Omega(m)}$ samples is necessary for learning the follower's private type from $\Theta = \emptyset \cup \{\theta: |\theta| = m/2 \text{ and } \theta \subset [m]\}$. 
We demonstrate the follower's type can be learned within a \textit{single} round with menus.

\begin{restatable}{restateprop}{hardnessExample}\label{prop:hardness_example}
    Within a single round, a menu of leader strategies is available to learn the follower's private type in Example \ref{example:peng2019}.
\end{restatable} 
\noindent In our proof, we construct a menu of leader strategies where every different follower type prefers different leader strategies. We refer the reader to the detailed construction of this menu in Appendix \ref{appendix:proof}.
\section{Conclusion}
Motivated by the recent research on the power of menus in contract design \cite{gan2022optimal,castiglioni2022designing,guruganesh2023power}, this paper put forwards a variant of the very basic online learning model in principal-agent problems by augmenting the principal policy space with posting a menu of principal strategies instead of a single principal strategy at every round of online interaction with the agent.
We provide an understanding of the learning complexities for a variety of problem settings, offering dedicated learning algorithms. 
Furthermore, we apply our general principal-agent framework to various concrete game instances, including widely studied cases like the contract design problem and the Stackelberg game, establishing a connection between our results and earlier online learning results in the literature.

There are many open questions around the power of menus in online principal-agent problems. For example,  
while our focus is to learn the agent's private type in this paper, it would be interesting to explore the problem of optimizing the principal utility throughout the online principal-agent interactions using menus.
What is more, our model assumes no constraints on the size of  principal menus. A natural question would be characterizing the sample complexity if the principal can only post a finite menu of strategies whose size is upper bound by a constant.
\clearpage
\section*{Acknowledgments}
This work is supported by NSF award CCF-2303372, Army Research Office
Award W911NF-23-1-0030 and Office of Naval Research Award N00014-23-1-2802.
\bibliography{refer}
\appendix 
\section{Omitted Proofs of Section \ref{sec:infinite_type}}\label{appendix:proof}
\infiniteType*
\begin{proof}
From equation \eqref{eq:rewrite_gradient}, we can learn the agent's utility parameter $\v^\theta_j$ up to a factor of $\lambda_j$ given a menu $\mmm_j = \{\x: |\x - \hat{\x}_j|^2\leq \epsilon \}$ for all $j \in [n]$. 
With the set of principal strategies $\{\hat{\x}_j\}_{j\in [n]}$ under assumption \ref{assum:oracle}, the principal can learn the agent's utility information $(\v^\theta_1, \cdots, \v^\theta_n)$ as $(\lambda_1 \tilde{\v}^\theta_1, \cdots, 
\lambda_n \tilde{\v}^\theta_n)$ with some unknown scalars $\bm{\lambda} = 
(\lambda_1,\cdots, \lambda_n)$. 

Next, we show how to learn $\bm{\lambda}$ given the knowledge of $(\tilde{\v}^\theta_1, \cdots,  \tilde{\v}^\theta_n)$, we ignore the superscript $\theta$ when it is clear from the context. Given any principal strategy $\x$ and its corresponding agent best response $j$, the agent's utility would be $V(\x, j) = \lambda_j \langle \tilde{\v}_j, \x \rangle + c_j$. Also, we have the agent's utility of responding with $j$ is greater than or equal to the utility of responding with any other $j' \ne j$, i.e., $\lambda_j \langle \tilde{\v}_j, \x \rangle + c_j \geq  \lambda_{j'} \langle \tilde{\v}_{j'}, \x \rangle + c_{j'}$. Formally, we have 
\[
\lambda_j \ge  \lambda_{j'} \frac{\langle\tilde{\v}_{j'}, \x \rangle}{\langle \tilde{\v}_j, \x \rangle} + \frac{c_{j'} - c_j}{\langle \tilde{\v}_j, \x \rangle}
,\quad \forall j'.
\]

Next, we introduce a useful lemma for finding a point on the separating hyperplane between the best response sub-regions of two actions $j, j'\in [n]$ (see our Figure \ref{fig:feasible_region} for an example of dividing the principal's strategy space into multiple best response sub-regions).
\begin{lemma}[Lemma 2, \cite{letchford2009learning}]\label{lem:letchford}
    Given two effective actions $j$ and $j'$ for the agent and the corresponding principal strategies $\hat{\x}_{j}$ and $\hat{\x}_{j'}$, with $O(L)$ samples, where $L$ is the representation precision, we can find a new point on the separating hyperplane between $j$ and $j'$ if the separating hyperplane exists.
\end{lemma}

We denote the principal strategy that is on the separating hyperplane between $j$ and $j'$ as $\hat{\x}_{j,j'}$. Note that a key property of $\hat{\x}_{j, j'}$ is that $V^\theta(\hat{\x}_{j, j'}, j) = V^\theta(\hat{\x}_{j, j'}, j')$, that is,
\begin{equation}\label{eq:hyperplane}
\lambda_j \langle \tilde{\v}_j, \hat{\x}_{j, j'} \rangle + c_j = \lambda_{j'} \langle \tilde{\v}_{j'}, \hat{\x}_{j, j'} \rangle + c_{j'}.
\end{equation}
where $\lambda_j, \lambda_{j'}, c_{j}$, and $c_{j'}$ are unknown variables. Thus, by sampling 4 different principal strategies $\hat{\x}_{j, j'}$ that are on the hyperplane, we can get 4 different equations in the format of \eqref{eq:hyperplane} and solve $\lambda_j, \lambda_{j'}, c_{j}$, and $c_{j'}$. With Lemma \ref{lem:letchford}, we can learn the values of $\lambda_j, \lambda_{j'}, c_{j}$, and $c_{j'}$ with the $O(4L)= O(L)$ samples, if there exists a separate hyperplane between the best response sub-regions of $j$ and $j'$.

Finally, there are at most $n$ convex regions since the agent has $n$ actions, there will be at most $\big(\begin{smallmatrix}n\\2\end{smallmatrix}\big)$ separating hyper-planes search as each region can share at most one hyperplane with the other region. Combining the hyper-planes search with Lemma \ref{lem:letchford} would directly imply the correctness of Theorem \ref{thm:infinite_type}.
\end{proof}

\hardnessExample*
\begin{proof}
Our proof is by construction, where we construct a menu that consists of a leader strategy from each type's best response region of $a^*$, i.e., $BR^\theta(a^*) = \{\x \in \xxx: \arg\max_{j\in [n]} V^\theta(\x, j) = a^*\}$. Specifically, for every type $\theta$ which corresponds to a specific set $\theta \subset [m]$ and $|\theta| = m/2$, we have any $\x$ in the best response region $BR^\theta(a^*) $ satisfy the following conditions \cite{peng2019learning}.
    \begin{equation*}
    \begin{split}
        \forall i \in \theta, x_i > \frac{2}{m} - \frac{m+2}{2N} - \frac{1}{N^2};  \,\,
        \forall i \notin \theta, x_i \le \frac{1}{N} + \frac{1}{N^2}.
    \end{split}
    \end{equation*}
    We propose to post a menu of strategies that consists of $\mmm=\{\x^\theta\}_{\theta \in \Theta}$ where $x^\theta_i = \frac{2}{m}$ if $i \in \theta$ and $0$ otherwise for every type $\theta \in \Theta$. 
    For follower type $\theta$ and for any leader strategy $\x^{\theta'} \in \mmm, \theta'\ne \theta$, we have 
    \begin{enumerate}
        \item $V^\theta(\x^{\theta'}, a_{m+1}) = -\frac{1}{N^3}$.
        \item $V^\theta(\x^{\theta'}, a_i) = -\frac{2}{m-2}$ if $i \notin \theta'$ and $V^\theta(\x^{\theta'}, a_i) =0$ otherwise.
        \item $V^\theta(\x^{\theta'}, a^*) \le max_{a \in \{a_i\}_{i\in [m+1]}}V^\theta(\x^{\theta'}, a)$ since $\x^{\theta'} \notin BR^\theta(a^*)$.
    \end{enumerate}
    On the other hand, we have $V^\theta(\x^\theta, a^*) = \frac{1}{N^2}$.  As a result, we got
    \begin{equation}\label{eq:effective_menu}
        V^\theta(\x^\theta, a^*) \geq \max_{\theta'} \max_{a} V^\theta(\x^{\theta'}, a)
    \end{equation}
    From equation \eqref{eq:effective_menu}, we can conclude that for every follower type $\theta$ with corresponding $|\theta| = m/2$, their optimal choice from the menu $\mathcal{X}=\{\x^\theta\}_{\theta \in \Theta}$ is $\x^\theta$ and then responds with $a^*$. However, if the follower chooses some leader strategy but responds with some $a \ne a^*$, we can conclude the follower's type is the type $\theta=\emptyset$.
\end{proof}

\end{document}